%% file: main.tex
\documentclass[a4paper,UKenglish]{article}
\usepackage{amssymb,mathtools}
\usepackage{amsfonts}
\usepackage{cite}
\usepackage[linesnumbered,lined,ruled,vlined,noend]{algorithm2e}
\usepackage{authblk}
\usepackage{lineno}
\usepackage{comment}
\usepackage{apxproof}
\usepackage{todonotes}
\input{macro}

\title{An Efficient Algorithm for Enumerating Chordal Bipartite Induced Subgraphs in Sparse Graphs}

\author[1]{Kazuhiro Kurita}
\author[2]{Kunihiro Wasa}
\author[1]{Hiroki Arimura}
\author[2]{Takeaki Uno}

\affil[1]{IST, Hokkaido University, Sapporo, Japan\\
  \texttt{\{k-kurita, arim\}@ist.hokudai.ac.jp}}
\affil[2]{National Institute of Informatics, Tokyo, Japan\\
  \texttt{\{wasa, uno\}@nii.ac.jp}}

\begin{document}
\maketitle
\begin{abstract}
    \input{abst}
\end{abstract}

\section{Intorduction}
\label{sec:intro}
A graph $G$ is \name{chordal} 
if any cycle with length four or more in $G$ has an edge, 
called a \name{chord}, 
which connects two nonconsecutive vertices on the cycle.  
If $G$ is a chordal, many NP-complete problems can be solved in polynomial time~\cite{Gavril:SIAM:1972}. 
Graph chordality is also interested in enumeration algorithm area. 
There are many efficient algorithms for enumerating subgraphs and supergraphs with chordality~\cite{Kiyomi:WG:2006,Kiyomi:IEICE:2006,Wasa:DS:2013,Daigo:SOFSEM:2009}. 
Moreover, chordality of a bipartite graph has been also well studied. 
A chordal bipartite graph is a bipartite graph without any induced cycles with length six or more. 
There are many characterizations of chordal bipartite graphs~\cite{Brandstadt:1999,Huang:JCT:2006,Uehara:ICALP:2002,Duris:IPL:2012}. 
In addition, the graph chordality is related to the hypergraph acyclicity~\cite{Beeri:JACM:1983,Brandstadt:1999}. 
In particular, Ausiello \etal  show that a hypergraph is $\beta$-acyclic if and only if its bipartite incident graph is chordal bipartite.

A subgraph enumeration problem is defined as follows: To output all subgraphs satisfying a constraint. 
To evaluate the efficiency of enumeration algorithms, we often measure in terms of the size of input and the number of outputs. 
An enumeration algorithm is \name{polynomial delay} if the maximum interval between two consecutive solutions is polynomial. 
Moreover, an enumeration algorithm is \name{amortized polynomial} if the total running time is $\order{M\cdot poly(N)}$ time, 
where $M$ is the number of solutions, $N$ is the input size, and $poly$ is a polynomial function. 
In enumeration algorithm area, 
there are efficient algorithms for sparse graphs~\cite{Wasa:Arimura:Uno:ISAAC:2014,Wasa:COCOON:2018,Alessio:Roberto:ICALP:2016,Eppstein:Loffler:EXP:2013,Kante:Limouzy:FCT:2011,Kazuhiro:ISAAC:2018}. 
Especially, \name{the degeneracy}~\cite{Lick:CJM:1970} of graphs has been payed much attention for constructing efficient enumeration algorithms. 

\noindent \textbf{Main results:}
In this paper, we propose chordal bipartite induced subgraph enumeration algorithm \EnumCB. 
In \EnumCB, we use a similar strategy as enumeration of chordal induced subgraphs. 
Kiyomi and Uno~\cite{Kiyomi:IEICE:2006} use a special vertex ordering, called the \name{perfect elimination ordering} $(v_1, \dots, v_n)$. 
In this ordering, any vertex $v_i$ is simplicial in $G[\setWithRange{V}{i}{n}]$, where $\setWithRange{V}{i}{n} = \inset{v_j \in V}{i \le j \le }n$.
Here, a vertex is called \name{simplicial} if an induced subgraph of neighbors becomes a clique. 
Kiyomi and Uno developed a constant delay enumeration algorithm for chordal induced subgraphs~\cite{Kiyomi:IEICE:2006} by using this ordering. 
Likewise a perfect elimination ordering of a chordal graph, we show that every chordal bipartite graph has a special vertex elimination ordering, 
and this is actually a necessary and sufficient condition of a chordal bipartite graph. 
We call this vertex ordering a \name{chordal bipartite elimination ordering} (\cbeo). 
\cbeo is defined by the following operation: Recursively remove a weak-simplicial vertex~\cite{Uehara:ICALP:2002}. 
Interestingly, \cbeo is a relaxed version of a vertex ordering proposed by Uehara~\cite{Uehara:ICALP:2002}. 
By using \cbeo, 
we propose our enumeration algorithm which outputs all chordal bipartite induced subgraphs in amortized $\order{kt\Delta^2}$ time,
where $k$ is the degeneracy of a graph, $t$ is the maximum size of $K_{t, t}$ as an induced subgraph, and $\Delta$ is the degree of $G$. 
Note that $t$ is bounded by $k$. 
Hence, \EnumCB enumerates chordal bipartite induced subgraphs in constant amortized time for bounded degree graphs.

\section{Preliminaries}
\label{sec:prelim}

Let $G = (V, E)$ be a simple graph, that is there is no self loops and multiple edges.
$u, v \in V$ are \name{adjacent} if there is an edge $\set{u, v} \in E$. 
The sequence of distinct vertices $\pi = (v_1, \dots, v_k)$ is a \name{path} if $v_i$ and $v_{i + 1}$ are adjacent for each $1 \le i \le k - 1$. 
If $v_1 = v_k$ holds in a path $C = (v_1, \dots, v_k)$, we call $C$ a \name{cycle}. 
\name{The distance} $dist(u, v)$ between $u$ and $v$ is defined by the length of a shortest path between $u$ and $v$.
We call a graph $H = (U, F)$ a \name{subgraph} of $G = (V, E)$ if $U \subseteq V$ and $F \subseteq E$ hold. 
A subgraph $H = (U, F)$ is an \name{induced subgraph} of $G$ if $F = \inset{\set{u, v} \in E}{u, v \in U}$ hold. 
In addition, we denote an induced subgraph as $G[U]$. 
\name{The neighbor of $v$} is the set of vertices $\inset{u \in V}{\set{u, v}\in E}$ and denoted by $N_G(v)$. 
If there is no confusion, we denote $N(v)$ as $N_G(v)$
In addition, we denote $N(v) \cap X$ as $N_X(v)$, where $X$ is a subset of $V$. 
The set of vertices $N[v] = N(v) \cup \set{v}$ is called \name{the closed neighbor}. 
We define \name{the neighbor with distance $k$} and \name{the neighbor with distance at most $k$} as 
$N^k(v) = \inset{u \in V}{dist(u, v) = k}$ and $N^{1:k}(v) = \bigcup_{1 \le i \le k} N(v)^k$, respectively. 
$\sig N(v)$ is \name{the neighbor set of neighbors} defined as $\inset{N(u)}{u \in N(v)}$. 
The \name{degree of $v$ $d(v)$} is the size of $N(v)$. 
The degree of a graph $G$ is the maximum size of $d(v)$ in $V$. 
Let $U$ be a subset of $V$. 
For vertices $u, v \in V$, $u$ and $v$ are \name{comparable} if $N(v) \subseteq N(u)$ or $N(v) \supseteq N(u)$ hold. 
Otherwise, $u$ and $v$ are \name{incomparable}. 
\begin{figure}[t]
    \centering
    \includegraphics[width=0.8\textwidth]{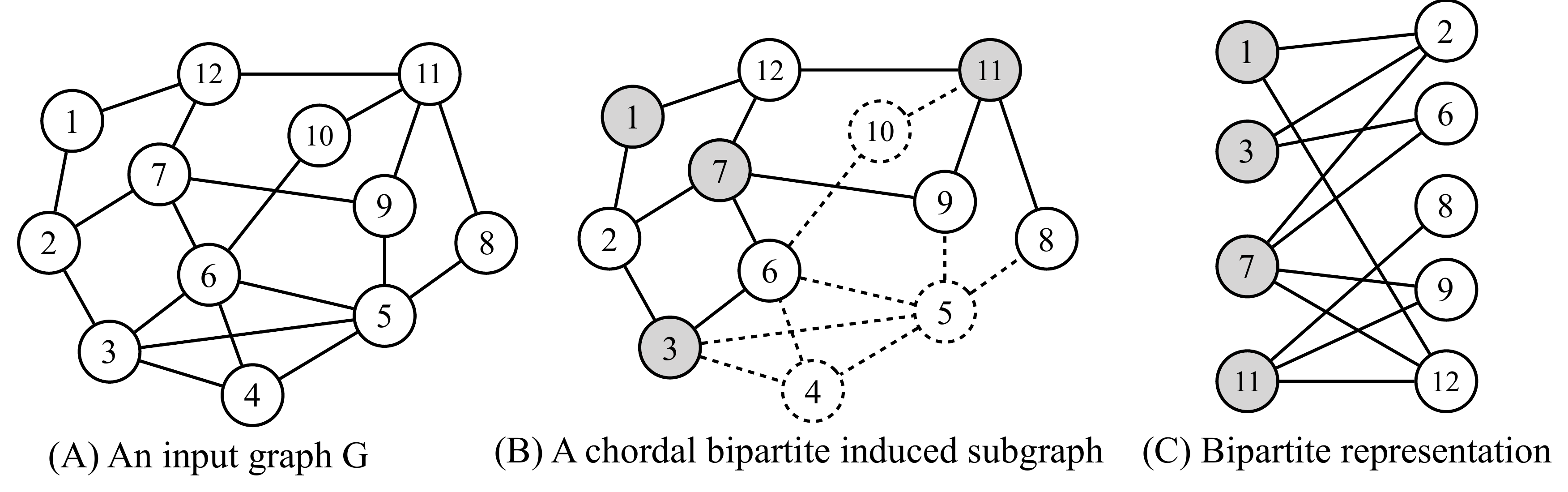}
    \caption{(A) shows an input graph $G$ and (B) shows one of the solutions $B = (X, Y, E)$, where $X = \set{1, 3, 7, 11}$ and $Y = \set{2, 6, 8, 9, 12}$. (C) shows the graph $B$ drawn by dividing $X$ and $Y$. }
    \label{fig:chordal}
\end{figure}
Let $B = (X, Y, E)$  be a bipartite graph. 
We call $B$ is \name{a chordal bipartite} graph if there is no induced cycles with the length four or more. 
A bipartite graph $B$ is \name{biclique} if any pair of vertices $x \in X$ and $y \in Y$ are adjacent. 
We denote a biclique as $K_{a, b}$ if $\size{X} = a$ and $\size{Y} = b$. 
In this paper, we consider only the case $a = b$ 
and the size of a biclique $K_{t, t}$ is $t$. 


Let $\sig{H} = (V, \sig E)$ be a \name{hypergraph}, 
where  $V$ is a set of vertices and $\sig E$ is a set of subsets of $V$. 
We call an element of $\sig E$ a \name{hyperedge}. 
For a vertex $v$, 
let 
$\sig H(v)$ be the set of edges $\inset{e \in \sig H}{v \in e}$ which contain $v$. 
A sequence of edges $C = (e_1, \dots, e_k)$ is a \name{berge cycle} if there exists $k$ distinct vertices $v_1, \dots, v_k$ such that 
$v_k \in e_1 \cap e_k$ and $v_i \in e_i \cap e_{i + 1}$ for each $1 \le i < k$. 
A berge cycle $C = (e_1, \dots, e_k)$ is a \name{pure cycle} if $k \ge 3$ and $e_i \cap e_j \neq \emptyset$ hold for any distinct $i$ and $j$, where $i$ and $j$ satisfy one of the following three conditions: (I) $|i - j| = 1$, (II) $i = 1$ and $j = k$, or (III) $i = k$ and $j = 1$. 
A cycle $C = (e_1, \dots, e_k)$ is a \name{$\beta$-cycle} if the sequence of $(e'_1, \dots, e'_k)$ is a pure cycle, where $e'_i = e_i \setminus \bigcap_{1 \le j \le k} e_j$. 
We call a hypergraph $\sig H$ \name{$\beta$-acyclic} if $\sig H$ has no $\beta$-cycles. 
We call a vertex $v$ a \name{$\beta$-leaf (or nest point)} if $e \subseteq f$ or $e \supseteq f$ hold for any pair of edges $e, f \in \sig H(v)$.
A bipartite graph $\sig I(\sig H) = (X, Y, E)$ is a \name{incidence graph} of a hypergraph $\sig H = (V, \sig E)$
if $X = V$, $Y = \sig E$, and $E$ includes an edge $\set{v, e}$ if $v \in e$, where $v \in V$ and $e \in \sig E$.
$\sig V$ is \name{totally ordered} if for any pair $X, Y \in \sig V$ of vertex subsets, 
either $X \subseteq Y$ or $X \supseteq Y$. 
we assume that $\sig{H}$ is not trivial, that is, $\sig{H}$ has more than one vertex.

Finally, we define our problem, chordal bipartite induced subgraph enumeration problem. 
In Fig.~\ref{fig:chordal}, we show an input graph $G$ and one of the solutions. 
\begin{problem}[Chordal bipartite induced subgraph enumeration problem]
\label{prob:def}
    Output all chordal induced subgraphs in an input graph $G$ without duplication. 
\end{problem}

\section{A Characterization of Chordal Bipartite Graphs}
\label{sec:char}

We propose a new characterization of chordal bipartite graphs. 
By using this characterization, we construct a proposed algorithm in Sect.~\ref{sec:algo}. 
We first give notions. 
A vertex $v$ is \name{weak-simplicial}~\cite{Uehara:ICALP:2002} 
if $N(v)$ is an independent set and any pair of neighbors of $v$ are comparable. 
A bipartite graph $B = (X, Y, E)$ is \name{bipartite chain} 
if any pair of vertices in $X$ or $Y$ are comparable, that is,  
$N(u) \subseteq N(v)$ or  $N(u) \supseteq N(v)$ holds for any $u, v \in X$ or $u, v \in Y$. 
To show the new characterization, we use the following two theorems.

\begin{theorem}[Theorem 1 of \cite{Ausiello:JCS:1986} ]
\label{theo:hyp}
    $\sig I(\sig{H})$ is chordal bipartite if and only if $\sig H$ is $\beta$-acyclic. 
\end{theorem}

\begin{theorem}[Theorem 3.9 of \cite{Brault:ACM:2016}]
\label{theo:brault}
A $\beta$-acyclic hypergraph $\sig H = (\sig V, \sig E)$ with at least two vertices has two distinct $\beta$-leaves that are not neighbors in $\sig{H'} = (\sig V, \sig E \setminus \set{\sig V})$.
\end{theorem}

Brault~\cite{Brault:ACM:2016} gives a vertex elimination ordering $(v_1, \dots, v_n)$ for a hypergraph $\sig H$, called a $\beta$-elimination ordering. 
The definition is as follows: 
For any $1 \le i \le n$, 
$v_i$ is a $\beta$-leaf in $\sig H[\setWithRange{V}{i}{n}]$, where $\setWithRange{V}{i}{n} = \inset{v_j \in V}{i \le j \le n}$. 
It is known that $\sig H$ is $\beta$-acyclic 
if and only if there is a $\beta$-elimination ordering of $\sig H$. 
Similarly, 
in this paper,
for any graph $G$, 
we define a vertex elimination ordering $(v_1, \dots, v_n)$ for $G$, called \cbeo, as follows: 
for any $1\le i \le n$,  $v_i$ is a weak-simplicial in $G[\setWithRange{V}{i}{n}]$. 
In the remaining of this section, 
we show that a graph is chordal bipartite if and only if there is \cbeo for $G$. 
Lemma~\ref{lem:leafweak} shows that a $\beta$-leaf of a hypergraph is weak-simplicial in its incident graph.  

\begin{lemma}
\label{lem:leafweak}
    Let $\sig H = (\sig V, \sig E)$ be a hypergraph, $v$ be a vertex in $\sig{V}$, 
    and $v'$ be the corresponding vertex of $v$ in $X$ of $\sig I(\sig H)$. 
    Then, $v$ is a $\beta$-leaf in $\sig H$ if and only if $v'$ is a weak-simplicial vertex in $\sig I(\sig{H})$. 
\end{lemma}

\begin{proof}
    We assume that $v$ is a $\beta$-leaf in $\sig H$. 
    Let $v'$ be the vertex corresponding to $v$ in $\sig I(\sig H)$. 
    From the definition of a $\beta$-leaf, 
    $\sig N(v)$ is also totally ordered in $\sig I(\sig H)$. 
    Thus, $v$ is a weak-simplicial vertex in $\sig I(\sig H)$. 

    We next assume that $v'$ is weak-simplicial in $X$ of $\sig I(\sig H)$. 
    From the definition, $\sig N(v)$ is totally ordered. 
    Thus, $\sig H(v)$ is totally ordered.
    Therefore, $v$ is a $\beta$-leaf in $\sig H$ and the statement holds. \qed
\end{proof}

From Lemma~\ref{lem:leafweak}, a $\beta$-leaf $v$ of $\sig H$ corresponds to a weak-simplicial vertex of an incident graph $\sig I(\sig H)$. 
We next show that a chordal bipartite graph has at least one weak-simplicial vertex from Theorem~\ref{theo:hyp}, Theorem~\ref{theo:brault}, and Lemma~\ref{lem:leafweak}. 




\begin{lemma}
\label{lem:twoweak}
    Let $B = (X, Y, E)$ be a chordal bipartite graph with at least two vertices. 
    If there is no vertex $v$ in $B$ such that $N(v) = X$ or $N(v) = Y$, 
    then $B$ has at least two weak-simplicial vertices which are not adjacent. 
\end{lemma}

\begin{proof}
    From Theorem~\ref{theo:hyp}, Theorem~\ref{theo:brault}, and Lemma~\ref{lem:leafweak}, 
    if $B$ is chordal bipartite and has no twins, then $G$ has at two weak-simplicial vertices which are not adjacent. 
    We now assume $B$ has twins. 
    We construct $B'$ as follows: 
    For each set of twins $T \in \sig T$, 
    remove all vertices in $T$ except one twin of $T$,  
    where $\sig T$ is the sets of twins of $B$. 
    Note that $B'$ is still chordal bipartite since vertex deletion does not destroy the chordality. 
    Since $B'$ has no twins, 
    $B'$ has at least two weak-simplicial vertices $u$ and $v$ which are not adjacent. 
    Since the set inclusion relation between $B$ and $B'$ is same, 
    $u$ and $v$ also weak-simplicial in $B$. 
    Hence, the statement holds. \qed
\end{proof}

\begin{theorem}
\label{theo:eo}
    Let $B$ be a bipartite graph. 
    $B$ is chordal bipartite if and only if there is a \cbeo for $B$. 
\end{theorem}

\begin{proof}
    From Lemma~\ref{lem:twoweak}, the only if part holds. 
    We consider the contraposition of the if part. 
    Suppose that $B$ is not chordal bipartite. 
    Then, $B$ has an induced cycle $C$ with length six or more. 
    Since a vertex in $C$ is not weak-simplicial, 
    we cannot eliminate all vertices from $B$ and the statement holds. \qed
\end{proof}

We next show that a vertex $v$ is weak-simplicial in a bipartite graph $B$ if and only if
$B[N^{1:2}(v)]$ is bipartite chain. 

\begin{lemma}
    \label{lem:ws:bchain}
    Let $B = (X, Y, E)$ be a chordal bipartite graph and $v$ be a vertex in $B$. 
    Then, $v$ is weak-simplicial if and only if an induced subgraph $B[N^{1:2}[v]]$ is bipartite chain. 
\end{lemma}

\begin{proof}
    We assume that $B[N^{1:2}[v]]$ is bipartite chain. 
    From the definition, any pair of vertices in $N(v)$ are comparable. 
    Hence, $v$ is weak-simplicial. 
    
    We next prove the other direction. 
    We assume that $v$ is weak-simplicial. 
    Let $x$ and $y$ be vertices in $N^2(v)$. 
    If $x$ and $y$ are incomparable in $B[N^{\le 2}[v]]$, 
    then there are two vertices $z \in N(x)\setminus N(y)$ and $z' \in N(y) \setminus N(x)$. 
    Note that $z$ and $z'$ are neighbors of $v$. 
    This contradicts that any pair of vertices in $N(v)$ are comparable. 
    Hence, $x, y \in N^2(v)$ are comparable and $B[N^{\le2}[v]]$ is bipartite chain. \qed
\end{proof}

\section{Enumeration of Chordal Bipartite Induced Subgraphs}
\label{sec:algo}

\begin{algorithm}[t]
\label{algo:propose}
  \caption{\texttt{ECB} enumerates all chordal bipartite induced subgraphs in
  amortized polynomial time.}
  \Procedure(\tcp*[f]{$G = (V, E)$: an input graph}){\EnumCB{$G$}}{
    \RecECB{$(\emptyset, \emptyset, V, G)$}\;
  }
  \Procedure{\RecECB{$X, WS(X), AWS(X), G$}}{
    Output $X$ and Compute $\cand{X}$ from $AWS(X)$\;
    \For{$v \in \cand{X}$}{
         $Y \gets X \cup \set{v}$\;
        \lIf{$\Parent{Y} = X$\label{algo:ws}}{
            \RecECB{$Y, WS(Y), AWS(Y), G$}\label{algo:candaws}
        }
    }
  }
\end{algorithm}

In this section, we propose an enumeration algorithm \EnumCB for chordal bipartite subgraphs based on reverse search~\cite{Avis:Fukuda:DAM:1996}. 
\EnumCB enumerates all solutions by traversing on a tree structure $\Fam{G} = (\sig S(G), \sig E(G))$, 
called a \name{family tree}, 
where $\sig S(G)$ is a set of solutions in an input graph $G$ and $\sig E(G) \subseteq \sig S(G) \times \sig S(G)$. 
Note that $\Fam{G}$ is directed. 
We give $\sig E(G)$  by defining the parent-child relationship among solutions based on Theorem~\ref{theo:eo}. 

Let $X$ be a vertex subset that induces a solution. 
We denote the set of weak-simplicial vertices in $G[X]$ as $WS(X)$.
In what follows, we number the vertex index from $1$ to $n$ and compare the vertices with their indices. 
The \name{parent} of $X$ is defined as $\Parent{X} = X \setminus \argmax{WS(X)}$. 
$X$ is a \name{child} of $Y$ if $\Parent{X} = Y$. 
Let $ch(X)$ be the set of children of $X$. 
We define the \name{parent vertex $pv(X)$} as $\argmax{WS(X)}$ which induces the parent. 
For any pair of solutions $X$ and $Y$, 
$(X, Y) \in \sig E(G)$ if $Y = \Parent{X}$. 
From Theorem~\ref{theo:eo}, any solution can reach the empty set by recursively removing the parent vertex from the solution.  
Hence, the following lemma holds. 

\begin{lemma}
\label{lem:ftree}
    The family tree forms a tree.
\end{lemma}
 
Next, we show that \EnumCB enumerates all solutions. 
For any vertex subset $X \subset V$, 
we denote $\setWithRange{X}{1}{i} = X \cap \setWithRange{V}{1}{i}$. 
An \name{addible weak-simplicial vertex set} is $AWS(X) = \inset{v \in V \setminus X}{v \in WS(X \cup \set{v})}$, 
that is, any vertex $v$ in $AWS(X)$ generates new solution $X \cup \set{v}$. 
We define a \name{candidate set} $\cand{X}$ as follows: 
$\cand{X} = \setWithRange{AWS(X)}{pv(X)}{n} \cup (AWS(X) \cap N^{1:2}(pv(X)))$. 
Note that $\cand{X}$ is a subset of $AWS(X)$. 
We show that the relation between $ch(X)$ and $\cand{X}$. 

\begin{lemma}
\label{lem:enum}
    Let $X$ and $Y$ be distinct solutions. 
    If $Y$ is a child of $X$, then $pv(Y) \in \cand{X}$. 
\end{lemma}
\begin{proof}
    Suppose that $Y$ is a child of $X$. 
    Let $v = pv(Y)$ and $u = pv(X)$. 
    Note that $v$ belongs to $AWS(X)$. 
    If $u < v$, then $v \in \setWithRange{AWS}{u}{n}(X)$ and thus $v \in \cand{X}$. 
    Otherwise, $u$ is not included in $WS(Y)$ since $v$ has the maximum index in $WS(Y)$. 
    From the definition of a weak-simplicial vertex, 
    there are two vertices in $N_{Y}(u)$ which are incomparable in $G[Y]$. 
    Since $\sig N(u)$ is totally ordered in $G[X]$, $v$ must be in $N^{1:2}(u)$. 
    Hence, the statement holds. \qed
\end{proof}

In what follows, 
we call a vertex $v \in \cand{X}$ \name{generates a child} if $X \cup \set{v}$ is a child of $X$. 
From Lemma~\ref{lem:ftree} and Lemma~\ref{lem:enum}, 
\EnumCB enumerates all solution by the DFS traversing on $\Fam{G}$. 

\begin{theorem}
\label{theo:enum}
    \EnumCB enumerates all solutions. 
\end{theorem}

\section{Time complexity analysis}
\begin{algorithm}[t]
  \caption{\texttt{UpdateWS} and \texttt{UpdateAWS} are efficient functions to compute $WS(Y)$ and $AWS(Y)$ by using $WS(X)$ and $AWS(X)$, respectively. }
  \Procedure{\UpdateWS{$X, v, WS(X), G$}}{
    $WS \gets WS(X)$\;
    \For{$u \in N_{X}(v)$}{
        \lIf{$u \in WS \land$ there is a vertex $w \in N_X(u)$ is incomparable to $u$. }{$WS \gets WS \setminus \set{u}$}
        \For{$w \in N_{X}(u) \cap WS$}{
            \lIf{$w$ and $v$ are incomparable}{
                $WS \gets WS \setminus \set{w}$
            }
        }
    }
    \Return $WS$\;
  }

  \Procedure{\UpdateAWS{$X, v, AWS(X), G$}}{
    $AWS \gets AWS(X)$\;
    \For{$u \in N(v)$}{
        \If{$u \in AWS$}{
            \lIf{There is a vertex $w \in N_X(u)$ which is incomparable to $u$. }{$AWS \gets AWS \setminus \set{u}$}    
        }
        \ElseIf{$u \in X$}{
            \For{$w \in N(u) \cap AWS$}{
                \lIf{$w$ and $v$ are incomparable. }{
                    $AWS \gets AWS \setminus \set{w}$
                }
            }
        }
    }
    \Return $AWS$\;
  }
\end{algorithm}

\begin{figure}[t]
    \centering
    \includegraphics[width=0.8\textwidth]{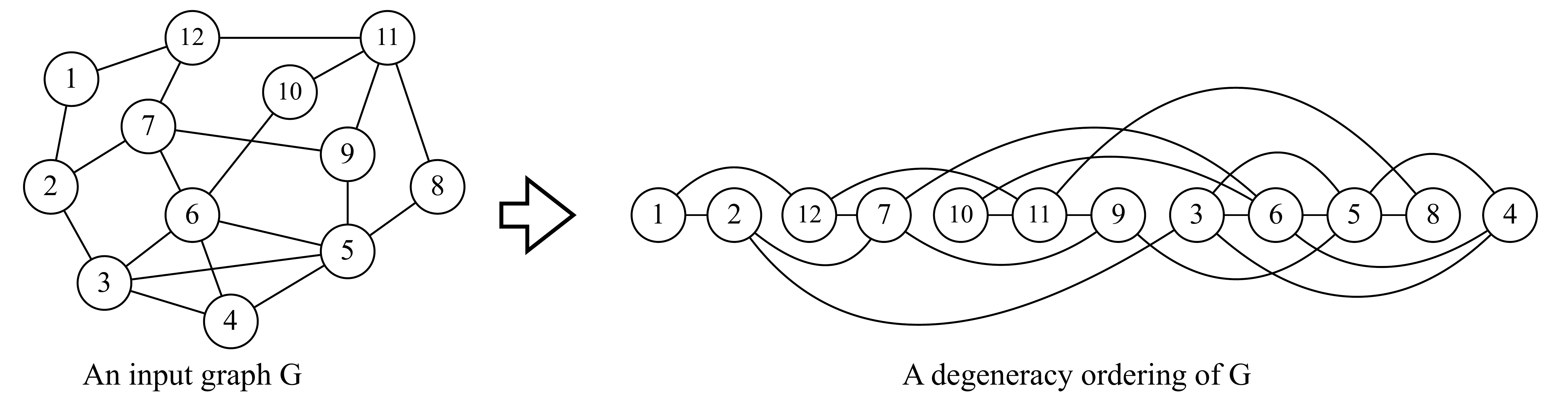}
    \caption{It is a degeneracy ordering of $G$. The degeneracy of $G$ is three. In this ordering, 
    $\size{\setWithRange{N^{\le2}}{1}{v}(v)}$ is at most $k\Delta$ for any vertex $v$. }
    \label{fig:dege}
\end{figure}
\EnumCB has two bottlenecks. 
(1) Some vertices in $\cand{X}$ does not generate a child and 
(2) the maintenance of $WS(X)$ and $AWS(X)$ consumes time.  
A trivial bound of the number of redundant vertices in $\cand{X}$ is $\order{\Delta^2}$ since 
only vertices in $(AWS(X) \cap N^{1:2}(pv(X))$ may not generate a child. 
To evaluate the number of such redundant vertices precisely, 
we use a \name{degeneracy ordering}. 
A graph $G$ is \name{$k$-degenerate} if any induced subgraph of $G$ has a vertex with degree $k$ or less~\cite{Lick:CJM:1970}. 
The \name{degeneracy} of a graph is the smallest such number $k$. 
Matula \etal~\cite{Matula:Beck:JACM:1983} show that 
a $k$-degenerate graph $G$ has a following vertex ordering: 
For each vertex $v$, the number of neighbors smaller than $v$ is at most $k$. 
This ordering is called a \name{degeneracy ordering} of $G$ (See Fig.~\ref{fig:dege}). 
Matula \etal also show that a linear time algorithm for obtaining one of the orderings. 
In what follows, 
we fix a reverse order of a degeneracy ordering and $WS(X)$ and $AWS(X)$ are sorted in this ordering. 
We first show that the number of redundant vertices is at most $2k\Delta$. 

\begin{lemma}
    \label{lem:fail:total}
    Let $X$ be a solution. 
    The number of vertices in $\cand{X}$ which do not generate a child is at most $2k\Delta$. 
\end{lemma}
\begin{proof}
    Let $v$ be a vertex in $\cand{X}$ and $p$ be a vertex $pv(X)$. 
    If $p < v$, then $v$ generates a child. 
    We assume that $v < p$. 
    Since $v$ is in $\cand{X}$, $v \in N^{1:2}(p) \cap \setWithRange{AWX}{1}{p}(X)$. 
    We estimate the size of $N^{1:2}(p) \cap \setWithRange{AWX}{1}{p}(X)$. 
    We consider a vertex $u \in \setWithRange{N}{1}{v}(v)$. 
    $\sum_{u \in \setWithRange{N}{1}{v}(v)}\size{N(u)}$ is at most $k\Delta$
    since $\size{\setWithRange{N}{1}{v}(v)}$ is at most $k$.
    We next consider a vertex $u \in \setWithRange{N}{v}{n}(v)$.
    Since $u$ is larger than $v$, a vertex in $\setWithRange{N}{u}{n}(u)$ is larger than $v$. 
    Hence, we consider vertices $\setWithRange{N}{1}{u}(u)$. 
    For each $u$, $\size{\setWithRange{N}{1}{u}(u)}$ is at most $k$. 
    Hence, $\sum_{u \in \setWithRange{N}{v}{n}(v)} \size{\setWithRange{N(u)}{1}{u}}$ is at most $k\Delta$ and the statement holds. 
    \qed    
\end{proof}

We next show how to compute $\cand{Y}$ from $\cand{X}$, where $X$ is a solution and $Y$ is a child of $X$. 
From the definition of $\cand{X}$, 
we can compute $\cand{Y}$ in $\order{\size{\cand{Y}} + k\Delta}$ time if we have $AWS(Y)$ and $pv(Y)$. 
Moreover, if we have $WS(X \cup \set{v})$, 
then we can determine whether $X \cup \set{v}$ is a child of $X$ or not in constant time since $WS(X \cup \set{v})$ is sorted. 
Hence, to obtain children of $X$, 
computing $AWS(X)$ and $WS(X)$ dominate the computation time of each iteration. 
Here, we define two vertex sets as follows: 
\begin{align*}
    \Delws{X}{v}  &= \inset{u \in N^{1:2}(v) \cap WS(X)}{u \notin WS(X \cup \set{v})} \text{ and }\\
    \Delaws{X}{v} &= \inset{u \in N^{1:2}(v) \cap AWS(X)}{u \notin WS(X \cup \set{u, v})}. 
\end{align*}
These vertex sets are the sets of vertices that are removed from $WS(X)$ and $AWS(X)$ after adding $v$ to $X$, respectively. 
In the following lemmas, we show that $WS(X)$ and $AWS(X)$ can be update if we have $\Delws{X}{v}$ and $\Delaws{X}{v}$. 

\begin{lemma} 
    \label{lem:ws}
    Let $X$ be a solution, $Y$ be a child of $X$, and $v = pv(Y)$. 
    Then, $WS(Y) = (WS(X) \setminus \Delws{X}{v}) \cup \set{v}$.  
\end{lemma}
\begin{proof}
    Let $u$ be a vertex in $WS(Y)$. 
    We prove $u$ is included in $(WS(X) \setminus \Delws{X}{v}) \cup \set{v}$. 
    If $u = v$ holds, then $u \in WS(Y)$ since $v \in pv(Y)$. 
    We assume that $u \neq v$. 
    Since $u$ is weak-simplicial in $Y$, $u \in WS(X)$. 
    If $u \in \Delws{X}{v}$, then $u$ is not weak-simplicial in $Y$. 
    It contradicts the assumption. Hence, $u \not\in \Delws{X}{v}$. 
    We prove the other direction. 
    Let $u$ be a vertex in $(WS(X) \setminus \Delws{X}{v}) \cup \set{v}$. 
    We assume that $u \neq v$. If $u \not\in WS(Y)$, then $u \in \Delws{X}{v}$. 
    It is contradiction and the statement holds. 
\end{proof}

\begin{lemma}
    \label{lem:aws}
    Let $X$ be a solution, $Y$ be a child of $X$, and $v = pv(Y)$. 
    Then, $AWS(Y) = AWS(X) \setminus \Delaws{X}{v}$.
\end{lemma}
\begin{proof}
    Let $u$ be a vertex in $AWS(Y)$. 
    We prove $u$ is included in $AWS(X) \setminus \Delaws{X}{v}$. 
    From the definition of $AWS(X)$, $u \in AWS(X)$ holds. 
    If $u \in \Delaws{X}{v}$, then $u$ is not weak-simplicial in $Y \cup \set{u}$. 
    Since $u \in AWS(Y)$, $u \not\in \Delaws{X}{v}$. 
    We prove the other direction. 
    Let $u$ be a vertex in $AWS(X) \setminus \Delaws{X}{v}$. 
    From the definition of $AWS(X)$ and $\Delaws{X}{v}$,
    $u$ is weak-simplicial in $X \cup \set{v, u}$. 
    Hence, $u \in AWS(Y)$ and the statement holds. \qed
\end{proof}

Note that by just removing redundant vertices, 
$WS(Y)$ and $AWS(X)$ can be easily sorted if $WS(X)$ and $AWS(X)$ were already sorted. 
We next consider how to compute $\Delws{X}{v}$ and $\Delaws{X}{v}$. 
We first show a characterization of a vertex in $\Delws{X}{v}$ and $\Delaws{X}{v}$. 
In the following lemmas, let $X$ be a solution, $v$ be a vertex in $AWS(X)$, and $Y$ be a solution $X \cup \set{v}$. 

\begin{lemma}
    \label{lem:delws:neighbor}
    Let $u$ be a vertex in $N_{Y}(v) \cap WS(X)$. 
    $u \in \Delws{X}{v}$ if and only if 
    $N_X(u)$ includes $w$ which is incomparable to $v$. 
\end{lemma}

\begin{proof}
    If $u$ has a neighbor $w$ in $X$ which is incomparable to $v$, 
    then from the definition,  $u$ is not weak-simplicial.
    
    Let $u$ be a vertex in  $\Delws{X}{v}$. 
    There is a pair of vertices $w_1$ and $w_2$ in $N_{Y}(u)$ which are incomparable. 
    If $w_1$ or $w_2$ is equal to $v$, then the statement holds.
    Hence, we assume that both $w_1$ and $w_2$ are not equal to $v$. 
    Since $G[Y]$ is bipartite, $w_1$ and $w_2$ are not adjacent to $v$. 
    Hence, $N_X(w_1) = N_Y(w_1)$ and $N_X(w_2) = N_Y(w_2)$ hold. It contradicts that $X$ is a solution and the statement holds. 
    \qed
\end{proof}

\begin{figure}[t]
    \centering
    \includegraphics[width=0.6\textwidth]{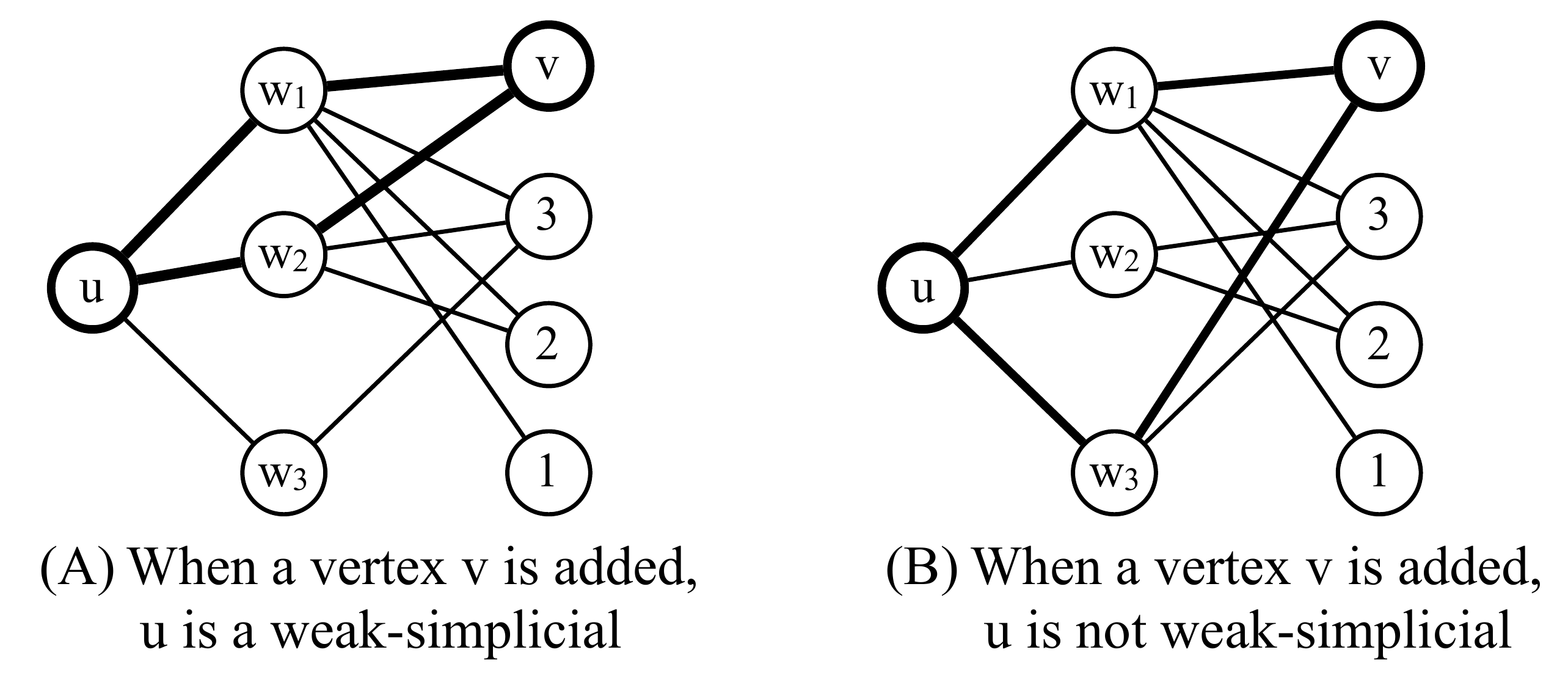}
    \caption{Let $X$ be a set of vertices $\set{u, w_1, w_2, w_3, 1, 2, 3}$. In $G[X]$, $L(X, u) = (w_1, w_2, w_3)$. In case (A), a vertex $u$ is still weak-simplicial. In case (B), however, $u$ is not weak-simplicial since $N(v)$ includes only $w_1$ and $w_3$.}
    \label{fig:dist2}
\end{figure}

\begin{lemma}
    \label{lem:delws:two}
    Let $u$ be a vertex in $N^2_{Y}(v) \cap WS(X)$. 
    Then, $u \in \Delws{X}{v}$ if and only if 
    there exists a pair of vertices $w_1, w_2 \in N_X(u)$ which satisfy 
    $N_X(w_1) \subset N_X(w_2)$, $v \in N_{Y}(w_1)$, and $v \not\in N_{Y}(w_2)$. 
\end{lemma}
\begin{proof}
    The if part is easily shown by the assumption of the incomparability of $w_1$ and $w_2$. 
    We next prove the other direction. 
    We assume that $u \in \Delws{X}{v}$. 
    Hence, there is a pair of neighbors $w_1$ and $w_2$ of $u$ such that they are incomparable in $Y$. 
    Without loss of generality, $N_{X}(w_1) \subset N_{X}(w_2)$ holds since $u$ is in $WS(X)$. 
    If $N_X(w_1) = N_X(w_2)$, then $w_1$ and $w_2$ are comparable in $Y$. 
    Since $w_1$ and $w_2$ are incomparable in $Y$, $w_1$ is adjacent to $v$ and $w_2$ is not adjacent to $v$. 
    Thus, the statement holds. \qed
\end{proof}

\begin{lemma}
    \label{lem:delaws:neighbor}
    Let $u$ be a vertex in $N(v) \cap AWS(X)$ and $Z = X \cup \set{u, v}$. 
    Then, $u \in \Delaws{X}{v}$ if and only if 
    $u$ has a neighbor $w$ in $Z$ which is incomparable to $v$. 
\end{lemma}
\begin{proof}
    The if part is trivial from the definition of weak-simplicial. 
    We prove the only if part. 
    We assume that $u \in \Delaws{X}{v}$ holds. 
    Since $u \in \Delaws{X}{v}$, 
    $u$ has a pair of neighbors $w_1$ and $w_2$ which are incomparable in $Z$. 
    If $w_1$ or $w_2$ is equal to $v$, then $u$ has a neighbor $w$ which is incomparable to $v$ and the statement holds.
    We next assume that $w_1$ and $w_2$ are distinct from $v$. 
    $v$ is not adjacent to $w_1$, $w_2$, or both of them since $G[Y]$ is bipartite. 
    Hence, $w_1$ and $w_2$ are comparable in $Z$ since $w_1$ and $w_2$ are comparable in $G[X]$. 
    This contradicts that $w_1$ and $w_2$ are incomparable in $Z$ and the statement holds. \qed
\end{proof}

\begin{lemma}
    \label{lem:delaws:two}
    Let $u$ be a vertex in $N^2(v) \cap AWS(X)$ and $Z = X \cup \set{u, v}$. 
    Then, $u \in \Delaws{X}{v}$ if  and only if 
    there exists a pair of vertices $w_1, w_2 \in N_{Z}(u)$ which satisfy 
    $N_Z(w_1) \subset N_{Z}(w_2)$,
    $v \in N_{Z}(w_1)$, and 
    $v \not\in N_{Z}(w_2)$.     
\end{lemma}
\begin{proof}
    From the assumption, $w_1$ and $w_2$ are incomparable in $Z$. 
    Hence, the if part holds. 
    We prove the other direction. 
    We assume that $u \in \Delaws{X}{v}$. 
    Hence, $u$ has a pair of vertices $w'_1$ and $w'_2$ which are incomparable in $Z$. 
    Without loss of generality, $N_{X\cup\set{u}}(w'_1) \subset N_{X\cup\set{u}}(w'_2)$ holds. 
    Since $w'_1$ and $w'_2$ are incomparable in $G[Z]$, $w'_1$ is adjacent to $v$. 
    Thus, the statement holds. \qed
\end{proof}


Next, we consider the time complexity of computing $\Delws{X}{v}$ and $\Delaws{X}{v}$. 
For analysing these computing time more precisely, 
we give two upper bounds with respect to the number of edges and the size of $N^2(v)$ in bipartite chain graphs. 
Note that $t$ is the maximum size of $K_{t,t}$ in $B$ as an induced subgraph. 

\begin{lemma}
    \label{lem:ws:vertices}
    Let $B$ be a bipartite chain graph and $v$ be a vertex in $B$. 
    Then, $\size{N^2(v)}$ is at most $\Delta$. 
\end{lemma}
\begin{proof}
    Let $u$ be a maximum vertex in $N(v)$ with respect to inclusion of neighbors. 
    Hence, for any vertex $w \in N(v)$, $N(w) \subseteq N(u)$ holds. 
    Since $N^2(v) = \bigcup_{w \in N(v)} N(w)$, $N^2(v)$ is equal to $N(u)$. 
    Hence, the statement holds. \qed
\end{proof}

\begin{lemma}
    \label{lem:ws:edges}
    Let $B = (X, Y, E)$ be a bipartite chain graph. 
    Then, the number of edges in $B$ is $\order{t\Delta}$, where $t$ is the maximum size of a biclique in $B$. 
\end{lemma}

\begin{proof}
    Let $v$ be a maximum vertex in $X$ with respect to inclusion of neighbors. 
    If $d(v) \le t$, then the statement holds from Lemma~\ref{lem:ws:vertices}. 
    We assume that $d(v) > t$. 
    We consider the number of edges in $B[N^{1:2}[v]]$. 
    Let $(u_1, \dots, u_{d(v)})$ be a sequence of vertices in $N(v)$ 
    such that $N(u_i) \subseteq N(u_{i+1})$ for $1 \le i < d(v)$. 
    For each $d(v) - t + 1 \le i \le d(v)$,  
    the sum of $\size{N(u_i)}$ is at most $\order{t\Delta}$. 
    We next consider the case for $1 \le i \le d(v) - t$. 
    Since $N(u_i)$ is a subset of $N(u_j)$ for any $i < j$, 
    $\size{N(u_i)}$ is at most $t$. 
    If $\size{N(u_i)}$ is greater than $t$, then $B$ has a biclique $K_{t + 1, t + 1}$. 
    Hence, the number of edges in $B$ is $\order{t\Delta}$ and the statement holds. \qed
\end{proof}


\begin{lemma}
    \label{lem:update:ws}
    Let $X$ be a solution, $v$ be a vertex in $\cand{X}$, and $Y$ be a solution $X \cup \set{v}$. 
    Then, we can compute $\Delws{X}{v}$ in $\order{t\Delta}$ time. 
\end{lemma}
\begin{proof}
    We first compute vertices in $WS(X)\cap N^2_X(v)$ that remain in $WS(Y)$. 
    From Lemma~\ref{lem:ws} and Lemma~\ref{lem:delws:two}, 
    $w \notin WS(Y)$ if and only if 
    there exists a pair of vertices $w_1, w_2 \in N(u)$ which satisfy 
    $N_X(w_1) \subset N_X(w_2)$, $v \in N_{Y}(w_1)$, and $v \not\in N_{Y}(w_2)$. 
    By scanning vertices with distance two from $v$,  
    this can be done in linear time in the size of 
    $\sum_{u \in N_X(v)} \size{N_X(u)}$. 
    Since $G[N^{1:2}(v)]$ is bipartite chain from Lemma~\ref{lem:ws:bchain},  
    $\sum_{u \in N_X(v)} \size{N_X(u)}$ is bounded by $\order{t\Delta}$ from Lemma~\ref{lem:ws:edges}.   Moreover, it can be determined whether $w \in N^2(v)$ and $v$ are comparable or not in this scan operation.

    We next compute vertices in $WS(X)\cap N_X(v)$ that remain in $WS(Y)$. 
    From Lemma~\ref{lem:delws:neighbor},
    $u$ is included in $\Delws{X}{v}$ if and only if $u$ has a neighbor $w$ which is incomparable to $v$. 
    Since $G[Y]$ is bipartite, $N_{X}(u)$ is included in $N^2_{Y}(v)$. 
    In the previous scan operation, we already know whether $w$ and $v$ are comparable or not. 
    Hence, we can compute whether $w \in WS(Y)$ or not in $\order{\size{N(u)}}$ time. 
    Since  $\order{\sum_{u \in N_{Y}(v)} \size{ N(u) } } = \order{t\Delta}$ holds from Lemma~\ref{lem:ws:edges}, 
    we can find $N_X(v) \cap \Delws{X}{v}$ in $\order{t\Delta}$ total time. 
    Hence, the statement holds\qed
\end{proof}

\begin{lemma}
    \label{lem:update:aws}
    Let $X$ be a solution and $v$ be a vertex in $\cand{X}$. 
    Then, we can compute $\Delaws{X}{v}$ in $\order{\Delta^2}$ time. 
\end{lemma}
\begin{proof}
    Let $Y$ be a solution $X \cup \set{v}$. 
    In the same fashion as Lemma~\ref{lem:update:ws}, 
    we can decide $u \in AWS(X \cup \set{v} \cap N^2_{Y}(v))$ in $\order{\Delta^2}$ total time. 
    By applying the above procedure for all vertices distance two from $v$, 
    we can obtain all vertices in $\Delaws{X}{v}$ in $\order{\Delta^2}$ time 
    since the number of edges can be bounded in $\order{\Delta^2}$. \qed
\end{proof}

From Lemma~\ref{lem:update:ws} and Lemma~\ref{lem:update:aws}, 
we can compute $\Delws{X}{v}$ and $\Delaws{X}{v}$ in $\order{t\Delta}$ and $\order{\Delta^2}$ time for each $v \in \cand{X}$, respectively. 
Hence, 
we can enumerate all children in $\order{\size{\cand{X}}t\Delta + \size{ch(X)}\Delta^2}$ time
from Lemma~\ref{lem:enum}, Lemma~\ref{lem:ws} and Lemma~\ref{lem:aws}. 
In the following theorem, we show the amortized time complexity and space usage of \EnumCB. 

\begin{theorem}
    \label{theo:time}
    \EnumCB enumerates all solutions in amortized $\order{kt\Delta^2}$ time by using $\order{n + m}$ space, 
    where $n$ is the number of vertices and $m$ is the number of edges in an input graph. 
\end{theorem}

\begin{proof}
    \EnumCB uses $AWS(X)$ and $WS(X)$ as data structures. 
    Each data structure demands linear space and the total space usage of \EnumCB is $\order{n + m}$ space. 
    Hence, the total space usage of \EnumCB is linear. 
    We consider the amortized time complexity of \EnumCB. 
    From Lemma~\ref{theo:enum}, \EnumCB enumerates all solutions. 
    From Lemma~\ref{lem:update:ws} and Lemma~\ref{lem:update:aws}, 
    \EnumCB computes all children and updates all data structures in $\order{\size{\cand{X}}t\Delta + \size{ch(X)}\Delta^2}$ time. 
    From Lemma~\ref{lem:fail:total}, $\size{\cand{X}}$ is at most $\size{ch(X)} + k\Delta$. 
    Hence, we need $\order{(\size{ch(X)} + k\Delta)t\Delta + \size{ch(X)}\Delta^2}$ time to generate all children. 
    Note that this computation time is bounded by $\order{(\size{ch(X)} + kt)\Delta^2}$. 
    We consider the total time to enumerate all solution. 
    Since each iteration $X$ needs $\order{(\size{ch(X)} + kt)\Delta^2}$ time, 
    the total time is $\order{\sum_{X \in \sig S}(\size{ch(X)} + kt)\Delta^2}$ time, 
    where $\sig S$ is the set of solutions. 
    Since $\order{\sum_{X \in \sig S}\size{ch(X)}\Delta^2}$ is bounded by $\order{\size{\sig S}\Delta^2}$, 
    the total time is $\order{\size{\sig S}kt\Delta^2}$ time. 
    Therefore, \EnumCB enumerates all solution in amortized $\order{kt\Delta^2}$ time and the statement holds. \qed
\end{proof}

\section{Conclusion}
\label{sec:conc}
In this paper, we propose a vertex ordering \cbeo 
by relaxing a vertex ordering proposed by Uehara~\cite{Uehara:ICALP:2002}. 
A bipartite graph $B$ is chordal bipartite if and only if $B$ has \cbeo, 
that is, this vertex ordering characterize chordal bipartite graphs. 
This ordering comes from hypergraph acyclicity and the relation between $\beta$-acyclic hypergraphs and chordal bipartite graphs. 
In addition, we also show that a vertex $v$ is weak-simplicial if and only if $G[N^{1:2}[v]]$ is bipartite chain. 
By using these facts, we propose an amortized $\order{kt\Delta^2}$ time algorithm \EnumCB. 

As future work, the following two enumeration problems are interesting: 
Enumeration of bipartite induced subgraph for dense graphs and
enumeration of chordal bipartite subgraph enumeration. 
For dense graphs, \EnumCB does not achieve an amortized linear time enumeration. 
If an input graph is biclique, then the time complexity of \EnumCB becomes $\order{nm}$ time. 
Hence, it is still open that there is an amortized linear time enumeration algorithm for chordal bipartite induced subgraph enumeration problem. 

\bibliographystyle{abbrv}
\bibliography{main.bbl}

\end{document}

%% file: macro.tex
\newcommand{\set}[1]{\{#1\}}
\newcommand{\inset}[2]{\{#1 \mid #2\}}
\newcommand{\name}[1]{{\emph{#1}}}

\newcommand{\size}[1]{\left|#1\right|}
\newcommand{\order}[1]{O(#1)}

\newcommand{\etal}{\textit{et al.}\xspace}

\newtheorem{lemma}{Lemma}

\newtheorem{theorem}[lemma]{Theorem}
\newtheorem{problem}{Problem}

\newcommand{\cand}[1]{C\left(#1\right)}

\newcommand{\sig}[1]{\mathcal{#1}}

\SetKwFunction{EnumCB}{ECB}
\SetKwFunction{RecECB}{RecECB}

\SetKwFunction{RecED}{RecED}%
\SetKwFunction{EICBR}{EICB2}%
\SetKwFunction{RecEICBR}{RecEICB2}%
\SetKwFunction{NextC}{NextC}%
\SetKwFunction{UpdateCand}{UpdateCand}%
\SetKwFunction{UpdateWS}{UpdateWS}%
\SetKwFunction{UpdateAWS}{UpdateAWS}%

\makeatletter
\@ifpackageloaded{todonotes}{

}{}
\makeatother

\makeatletter
\@ifpackageloaded{algorithm2e}{
\SetKwProg{Fn}{Function}{}{}
\SetKwProg{Procedure}{Procedure}{}{}
\SetKwProg{Subprocedure}{Subprocedure}{}{}
\SetKwComment{tcc}{//}{}
\SetKwFunction{Output}{Output}%
\SetKw{Continue}{continue} 
\SetKwInOut{AlgInput}{Input}
\SetKwInOut{AlgOutput}{Output}
\SetKwInOut{AlgPrecondition}{Pre-conditions}
\SetKwInOut{AlgInvariant}{Invariants}
}{}
\makeatother

\newcommand{\Fam}[1]{\mathcal{F}\left(#1\right)}
\newcommand{\Parent}[1]{\mathcal{P}\left(#1\right)}

\newcommand{\argmax}{\mathop{\rm arg~max}\limits}
\newcommand{\cbeo}{\texttt{CBEO}\xspace}
\newcommand{\Delws}[2]{Del_\text{W}\left( #1, #2 \right)}
\newcommand{\Delaws}[2]{Del_\text{A}\left( #1, #2 \right)}
\newcommand{\setWithRange}[3]{#1_{#2:#3}}

%% file: abst.tex
In this paper, we propose a characterization of chordal bipartite graphs and an efficient enumeration algorithm for chordal bipartite induced subgraphs. 
A chordal bipartite graph is a bipartite graph without induced cycles with length six or more. 
It is known that the incident graph of a hypergraph is chordal bipartite graph if and only if the hypergraph is $\beta$-acyclic. 
As the main result of our paper, 
we show that a graph $G$ is chordal bipartite if and only if there is a special vertex elimination ordering for $G$, called \cbeo. 
Moreover, we propose an algorithm \EnumCB which enumerates all chordal bipartite induced subgraphs in $O(kt\Delta^2)$ time per solution on average, 
where $k$ is the degeneracy, $t$ is the maximum size of $K_{t,t}$ as an induced subgraph, and $\Delta$ is the degree. 
\EnumCB achieves constant amortized time enumeration for bounded degree graphs.
